\documentclass[12pt]{article}
\usepackage{amsmath,amsfonts,amssymb,amsthm}
\usepackage{url,setspace}
\newcommand{\phm}{\phantom{-}}
\newcommand{\barz}{\overline{z}}
\newcommand{\barzx}{\overline{zx}}
\newcommand{\barzxx}{\overline{zx^2}}
\newcommand{\wt}{\widetilde}

\newcommand{\tp}{\tilde p}

\usepackage[round]{natbib}
\usepackage{graphicx}
\graphicspath{{./Figs/}}
\usepackage{booktabs}
\usepackage[margin=1in]{geometry}

\newcommand\real{\mathbb{R}}
\newcommand\tran{\mathsf{T}}

\newcommand{\e}{\mathbb{E}}
\newcommand{\var}{\mathrm{Var}}
\newcommand{\cov}{\mathrm{Cov}}
\newcommand{\err}{\varepsilon}
\renewcommand{\ge}{\geqslant}
\renewcommand{\le}{\leqslant}

\newcommand{\cx}{\mathcal{X}}

\newcommand{\rd}{\,\mathrm{d}}
\newcommand{\dnorm}{\mathcal{N}}
\newcommand{\dunif}{\mathcal{U}}
\newtheorem{proposition}{Proposition}
\newcommand{\giv}{\!\mid\!}
\newcommand{\pre}{\mathrm{pre}}
\author{Art B. Owen\footnote{Art Owen is a professor at Stanford
University.  Most of the work reported here was done for Google and
was not part of his Stanford responsibilities.}
\\Stanford University \and Hal Varian\footnote{Hal Varian is Chief Economist at Google.} \\Google Inc.}

\title{Optimizing the tie-breaker regression discontinuity design}
\date{July 2020}
\begin{document}
\maketitle

\begin{abstract}
  Motivated by customer loyalty plans and scholarship programs, we
  study tie-breaker designs which are hybrids of randomized controlled
  trials (RCTs) and regression discontinuity designs (RDDs).  We
  quantify the statistical efficiency of a tie-breaker design in which
  a proportion $\Delta$ of observed subjects are in the RCT. In a two
  line regression, statistical efficiency increases monotonically with
  $\Delta$, so efficiency is maximized by an RCT.  We point to
  additional advantages of tie-breakers versus RDD: for a
  nonparametric regression the boundary bias is much less severe and
  for quadratic regression, the variance is greatly reduced.  For a
  two line model we can quantify the short term value of the treatment
  allocation and this comparison favors smaller $\Delta$ with the RDD
  being best.  We solve for the optimal tradeoff between these
  exploration and exploitation goals.  The usual tie-breaker design
  applies an RCT  on the middle $\Delta$ subjects as ranked by the assignment
  variable.  We quantify the efficiency of other designs such as
  experimenting only in the second decile from the top.  We also show
  that in some general parametric models a Monte Carlo evaluation can
  be replaced by matrix algebra.
\end{abstract}

% keywords at EJS
%electronic commerce; experimental design;  hybrid experiments
% 62K99 (design)  62J99 (regression) 62F99 (parametric)

\section{Introduction}

Airlines, hotels and other companies
may offer incentives such as free upgrades
to their most loyal customers.
An e-commerce company may offer some analytic tools
or other support to the customers most likely to benefit from them.
A philanthropist may offer higher education scholarships
to high school students with excellent GPAs.
It is reasonable to expect some benefit from the subjects
who receive the treatment, be it increased sales to a customer
or better educational outcomes for a student.
It is then of interest to measure the causal effect of these special treatments.
A natural choice in this context is the regression discontinuity
design (RDD) but that has the disadvantage of only estimating
a causal impact right at the threshold point separating
treated from untreated study subjects.

In this paper we study a tie-breaker design that can estimate
the causal effect more broadly. That design injects some randomness
into the decision near the cutoff.   Our main contributions
are to analyze the efficiency gains of tie-breaker designs versus
RDD, and to study
the tradeoffs behind deciding how much randomness to introduce.
More randomness brings greater statistical efficiency, while at the same
time, it is expected to reduce the value of the incentives by not
applying them where they will be the most effective.

The RDD was originated by \citet{this:camp:1960}.
In an RDD, subjects are sorted according to a treatment assignment variable $x$
and those for which $x$ exceeds a threshold $t$ get the
treatment while others do not.
Sometimes the assignment  variable is called a running variable
or a forcing variable.
For background on RDD see
\citet{Angrist09,Angrist14},
\citet{Imbens08}, \citet{Bloom08}, \citet{Kaauw08}, and
\citet{Lee10}.

Historically, regression discontinuity designs were fit by a
regression including a polynomial in $x$ and a discontinuous predictor
$1\{x>t\}$ whose coefficient was taken to be the estimated causal
impact of the treatment at $x=t$.  This approach is problematic.  Low
order polynomial estimates are biased by lack of fit, and high order
ones are unstable \citep{gelman2017high}.  The more modern approach
fits linear or quadratic or other low order local polynomial
regression models to the left and right of the threshold using kernel
weights proportional to $K(|x-t|/b)$ for a bandwidth $b>0$ and a
kernel $K(\cdot)$ of bounded support.
In the kernel regression approach,
the estimated causal effect is the difference between those nonparametric
regressions when extrapolated to $x=t$ using data sets on the left and
the right of $t$.  See~\cite{hahn2001identification} for a description,
\cite{porter2003estimation} for optimality results, and
\cite{calonico2014robust} for improved confidence interval estimation.
\cite{armstrong2018optimal} and \cite{imbens2019optimized} optimize
for regression functions in general convex classes while also taking
special care with assignment variables that have a discrete distribution.

One problem with RDDs is that a causal estimate is only
available at $x=t$.  A randomized controlled trial (RCT)
by contrast makes the treatment a random variable
independent of $x$.  An RCT  would not be appropriate for
a customer loyalty program, and even less so for a scholarship

This problem is well suited to a tie-breaker design.  For an assignment
variable $x$, subjects are assigned to a
control condition if $x\le A$, to a test condition if $x\ge B$ and
their treatment (test or control) is randomized if $A< x< B$.  If
$A=B$, then no subjects are randomized and the design is an RDD.  At
the other extreme, if all the $x$ values are between $A$ and $B$, then
the design is an RCT as described in texts on causal inference
\citep{imbens2015causal} or on experimental design
\citep{box1978statistics,wu2011experiments}.  Tie-breaker designs are
also called cutoff designs; see \citep{CappelleriXX}.  If one is fitting
kernel weighted regressions, then the tie-breaker design offers an
additional advantage.  Nonparametric regressions have their most severe bias
problems at or outside the boundary of the observed data
\citep{rice1983smoothing}.  When $A<B$ there is a whole interval of
$x$ values for which the nonparametric regressions need not
extrapolate.

%Sometimes we refer to subjects
%getting the treatment or not, in place of getting test and control levels
%of the treatment.

\cite{Angrist2014leveling} use a tie-breaker design to evaluate the effects of
post secondary aid in Nebraska. In that setting, $x$ was a student ranking.
Students were triaged into top, middle and bottom groups.
The top students received aid, the bottom ones did not,
and those in the middle group were randomized to receive aid or not.
\cite{aike:west:schw:carr:hsiu:1998} report on a study
about allocation of students to remedial English classes where
the assignment variable is a measure of students' reading ability before they
matriculate.

Tie-breakers are not the only settings where the threshold
varies.  In fuzzy RDDs \citep{camp:1969} the threshold varies due to
dependence on other variables that may be unavailable to the
data analyst.
The threshold can also vary in settings where subjects or others working on their
behalf manipulate the value of $x$ in order to get the treatment
\citep{mccrary2008manipulation}. \cite{rosenman2019optimized}
propose a mitigation strategy.
We focus on the tie-breaker setting because in our motivating
problems the investigator has control of the treatment variable.

Our interest is in optimizing the size of the RCT within a tie-breaker
experiment.  For this purpose we need a statistical model.  In
Section~\ref{sec:general} we give a very general approach to this
problem but it provides no closed form results.  For interpretable
results, we work primarily with a model in which there are two linear
regressions, one for treatment and one for control.  This model was used
by \cite{gold:1972} and \cite{Bloom08}, who both find RCTs
more efficient than RDDs, and we think it is the simplest one in which
the tradeoff we study is interesting. We can interpolate between RCTs
and RDDs using a quantity $\Delta\in[0,1]$ representing the fraction
of experimental assignments, ranging from $\Delta=0$ for the RDD to
$\Delta=1$ for the RCT.  When the region of study is small then a
linear model will perform similarly to the local linear models
underlying kernel approaches to RDD.  At the design stage we know a
lot less about model goodness of fit than we will once the data are
available and that is another reason to design with a simple working
model.

Figure~\ref{fig:thick}  illustrates tie-breaker
designs for four values of $\Delta$. The assignment variable
there has a Gaussian distribution, that we assume has
been centered and scaled. The outcome variable is simulated
from a linear model with a constant treatment effect.
For instance, in the third panel, the top $1/6$
of subjects get the treatment, the bottom $1/6$ do not
and a fraction $\Delta = 2/3$ of the
data in the middle have randomized allocation.

For a Gaussian assignment variable, the experimental
region in the middle of the data is where the data are
most densely packed, which may well be where we
are most interested in learning the treatment effect.
The effect of treatment appears to be more visually
prominent at larger $\Delta$ in accordance
with the greater statistical efficiency that we
find here for larger~$\Delta$.

\begin{figure}[t]
\centering
\includegraphics[width=.9\hsize]{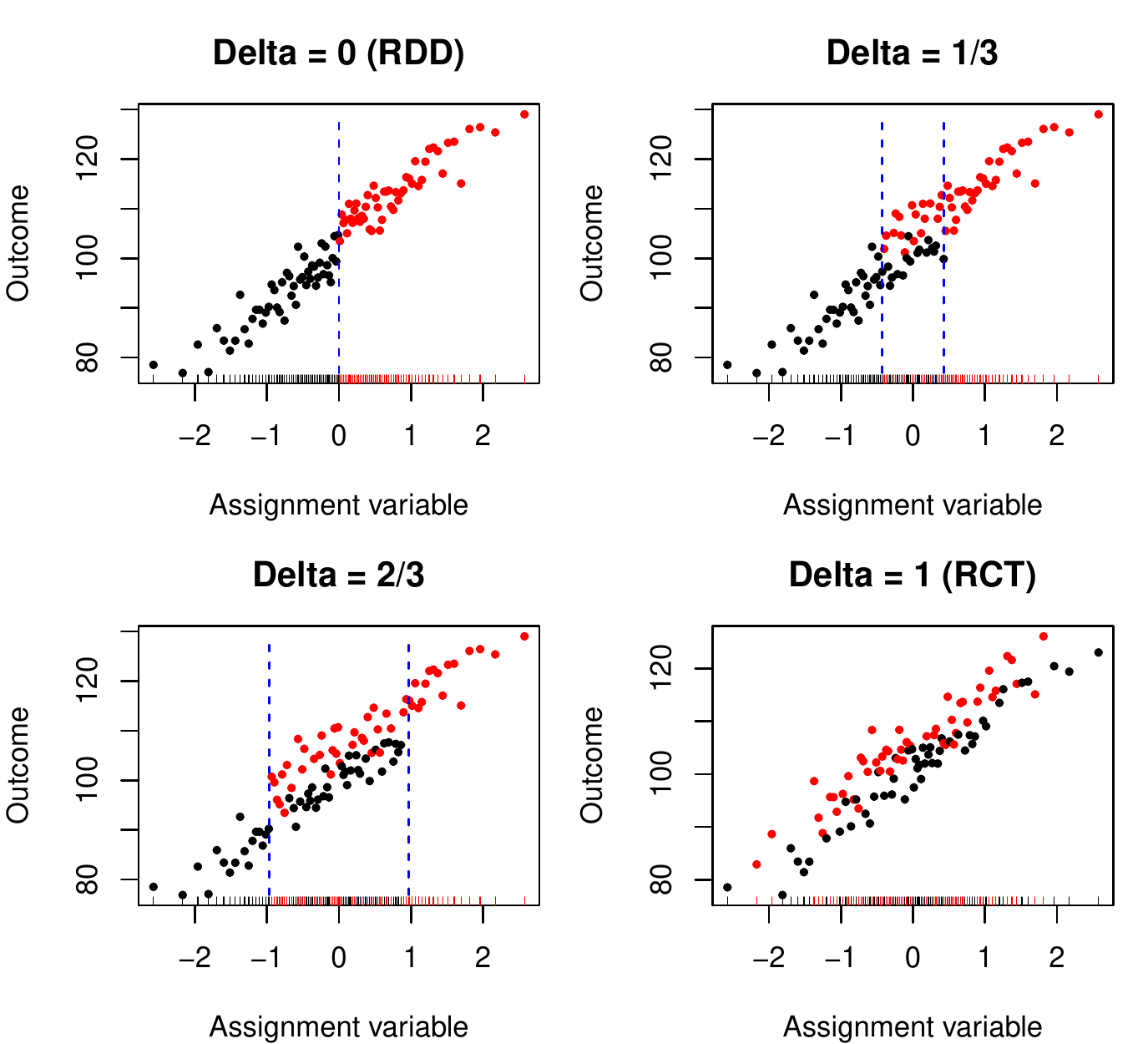}
\caption{Illustrative data for tie-breaker designs with $\Delta\in\{0,1/3,2/3,1\}$,
  and a standardized Gaussian assignment variable.
  The regression discontinuity design has $\Delta=0$,
  the randomized controlled trial has $\Delta=1$.
  Treated points are plotted in red, control in black.
  Allocation is deterministic for $x$ outside the blue lines.
}
\label{fig:thick}
\end{figure}

Our first working model is a two-line regression relating an outcome
to a uniformly distributed assignment variable. The early sections of our
paper work in this framework.  Section~\ref{sec:setup} introduces that
working model.  The slope and intercept vary between treatment and
control.  Section~\ref{sec:efftwo} shows that the statistical
efficiency of incorporating $\Delta>0$ experimentation versus the
plain regression discontinuity design at $\Delta=0$ is
${1+3\Delta^2(2-\Delta^2)}$, when $x\sim \dunif[-1,1]$.  Thus,
statistical efficiency is a monotone increasing function of the amount
of experimentation. At the extreme, a pure RCT with $\Delta=1$ is $4$
times as efficient as the RDD as was found earlier by
\cite{Bloom08}.  We ordinarily expect that our outcome
variable will show the greatest gains if we give the treatment to the
highest ranked subjects and a tie-breaker design will then reduce
those gains.  Section~\ref{sec:cost} quantifies that cost in the
two-line regression model and trades it off against statistical
efficiency. The optimal $\Delta$ is then dependent on the ratio
between the value per subject of the short term return and the value
of the information per subject that we get for a given $\Delta$.  It
can be hard to know how  much to weigh the information value
compared to the short term value.  A practical option is to choose the
smallest $\Delta$ with at least a specified amount of efficiency with
respect to the full experiment $\Delta=1$.
%Although an experiment might be designed for a linear model, once the
%data are collected there may be nonlinearities that warrant a more
%flexible model.

The later sections of our paper generalize beyond the working model.
Section~\ref{sec:quadratic} repeats our analysis of the linear model
for a pair of quadratic regression models.  We see that the regression
discontinuity design has a much higher variance than the experiment
does, in line with the instability findings of \cite{gelman2017high}
mentioned previously.  An RCT can have orders of magnitude less
variance than the RDD with tie-breaker designs in between.
Section~\ref{sec:gaussian} handles the case of a Gaussian assignment
variable that we illustrate in Figure~\ref{fig:thick}. It is similar
to the uniform case. Here a full RCT is $\pi/(\pi-2)\doteq 2.75$ times
as efficient as the RDD as was found by \citet{gold:1972}.

Section~\ref{sec:sliding} looks at replacing the three treatment
probabilities $0$\%, $50$\% and $100$\% by a strategy with more levels
or even a continuous sliding scale $p(x)$ of the assignment variable $x$.
We show that there is little to gain by this.  If $p(x)$ satisfies
$p(-x) = 1-p(x)$, as with  a symmetric CDF, then in a two line model
both the information gained and the value from the experimental subjects in any sliding scale
can also be attained by tie-breaker design using only levels $0\%$,
$50\%$ and $100$\%.  A non-symmetric sliding scale can be symmetrized
without affecting its cost and potentially reducing the variance of
some of the regression coefficients.  Using treatment probabilities
$\epsilon$, $0.5$ and $1-\epsilon$ would not improve efficiency but
would allow a potential outcomes analysis \citep{imbens2015causal} of
the data.

Section~\ref{sec:general} describes a numerical version of our
approach that does not require a simplistic regression model and allows
users to choose their own.  The design can then be chosen by an
intensive numerical search with a Monte Carlo evaluation of each
design choice.  We show how to replace that simulation-based inner
loop by matrix algebra allowing faster and more thorough optimization.
We also find in Section~\ref{sec:general} that experimenting on all data maximizes statistical
efficiency in very general circumstances.

The tie-breaker literature
has emphasized experiments in the middle range of the assignment variable
$x$.  Section~\ref{sec:noncentral} looks at off center experiments,
such as experimenting in just the second decile from the top.  In our
motivating applications, the treatment might only be offered to a
small fraction of subjects.  Experimenting in the second decile
reduces the most important regression coefficient's variance to about
60\% of what it would be with a comparably sized regression
discontinuity design.  When planning the experiment we don't know
whether the linear models that might work on the highest ranked
subjects would hold for all subjects.  At the time of analysis, we
might opt to reduce a bias by only including the highest ranked 30\%
of subjects in the analysis. This is similar to a kernel weighting.
Reducing the data set that way would greatly increase the variance of
both the RDD and tie-breaker designs. Interestingly, in this example,
the efficiency ratio between the two approaches is almost unchanged.
Section~\ref{sec:discussion} contains a short discussion of how to use
the findings.

We close this introduction with an historical note.  In the
Lanarkshire milk experiment, described by \cite{stud:1931} the goal
was to measure the effect of a daily ration of milk on the health of
school children. Among many complications was the fact that some of
the schools chose to give the rations to the students that they
thought needed it most.  While that may have been the most beneficial
way to allocate the school's milk, it was very damaging to the process
of learning the causal impact of the milk rations. A tie-breaker
experiment might have been a good compromise.

%It is well known in the literature that experiments
%are more efficient than regression discontinuity designs.
%Section 6 of \citet{Bloom08}
%  discusses this point in depth.
%They include the four-fold efficiency improvement
%we get for uniformly distributed assignment variables
%and a factor of $2.75$ for normal assignment variables.
%The latter goes back to \citet{gold:1972}.

\section{Setup}\label{sec:setup}

We begin with a simple setting where there are an even number $N$ of
subjects $i=1,\dots,N$, and exactly $N/2$ of them will receive the
treatment.  There is an assignment variable $x_i\in\real$
for which it is reasonable to give the treatment to subjects
with the largest values.
The assignment variable might be the output of a
statistical machine learning model based on multiple variables,
or it could be based on a subjective judgment of one or more
experts or stakeholders.

We will simplify the problem by transforming $x_i$ to be equispaced in
the interval $[-1,1]$.  That is, after sorting the subjects into
increasing order of $x_i$, we make a rank transformation to
$x_i = (2i-N-1)/N$.
Let $z_i$ indicate the treatment status;
subjects that receive the treatment have $z_i=+1$ and subjects
that do not receive the treatment have $z_i=-1$.

We denote the experimental interval by $(-\Delta,+\Delta)$ for
$\Delta$ in $[0,1]$.  In our hybrid design the treatment assignment $z_i\in\{-1,1\}$
includes some randomization as follows:
\begin{align}\label{eq:hybridz}
\Pr(z_i=1\giv x_i) =
\begin{cases}
\,\ \ 1, & x_i \ge\Delta\\
1/2, & |x_i|<\Delta\\
\,\,-1, & x_i \le -\Delta.
\end{cases}
\end{align}

If $\Delta=0$, then we have a classic RDD with the
discontinuity at $x=0$.  If $\Delta=1$, then we have a classic RCT.
If $0<\Delta<1$, then we have a tie-breaker design
with $\Delta$ measuring the amount of randomization.

The random allocation in equation~\eqref{eq:hybridz} will, on average, make half of
the $z_i$ for $|x_i|<\Delta$ equal $1$ and the other half equal $-1$.
One way to do this is to choose $z_i=1$ for a simple random
sample of half of the elements in $R=\{i\mid |x_i|<\Delta\}$.
Stratified schemes, setting $z_i=1$ for exactly one random member of
each consecutive pair of indices in $R$ are also easy to implement.

The impact of the treatment is measured by a scalar outcome $Y$ where
$Y_i$ is a measure of the benefit derived from subject $i$.
That could be future sales in a commercial setting or a measure
of post-secondary educational success for a scholarship.
We suppose that the delay time between setting $z_i$ and observing
$Y_i$ is long enough to make bandit methods (see for instance,
\cite{scott2015multi}) unsuitable.  We will instead compare
experimental designs using the following two-line regression model:
\begin{align}\label{eq:twolines}
%\e(Y) = \beta_0+\beta_1x+\beta_2z+\beta_3xz.
Y_i = \beta_0+\beta_1x_i+\beta_2z_i+\beta_3x_iz_i + \err_i,
\end{align}
where $\err_i$ are IID random variables with mean $0$ and
finite variance $\sigma^2>0$.
Our analysis is based on the regression model~\eqref{eq:twolines}
instead of the randomization because the treatment for subjects
with $x$ outside $(-\Delta,\Delta)$ is not random.
See Section~\ref{sec:sliding} for an alternative.

The effect of the treatment averaged over subjects $i=1,\dots,N$ is
$2\beta_2$. The factor of $2$ comes from comparing $z_i=1$ to
$z_i=-1$.  We can also estimate whether the effect increases or
decreases with $x$, through the coefficient $\beta_3$.  The
quantity $2\beta_2$ is also the magnitude of the treatment effect on
a (hypothetical) average subject with $x=0$.

Under model~\eqref{eq:twolines}, we can distinguish subjects for whom
the treatment is effective from those for whom it is not.
Suppose that $\tau$ is the incremental cost of offering the treatment to one
 subject.  This might be a support cost or foregone revenue;
in an educational context it would be the cost of offering a scholarship.
If $\beta_3>0$, then there is a cutpoint
$$x_*=\frac{\tau-2\beta_2}{2\beta_3}$$ with
$\e(Y\giv z=1)-\e(Y\giv z=-1)\ge \tau$ for subjects with $x\ge x_*$.
If $x_*>1$ then the treatment does not pay off for any subject
while if $x_*<1$ then it pays for all subjects.
If $\beta_3<0$, then
the treatment only pays off for subjects with $x_i \le x_*$.
We discuss that case further in Section~\ref{sec:cost}.

\section{Efficiency in the two-line model}\label{sec:efftwo}

We will analyze the data $(x_i,Y_i)$ for $i=1,\dots,N$ by fitting
model~\eqref{eq:twolines} by least squares.  The parameter of interest
is $\beta=(\beta_0,\beta_1,\beta_2,\beta_3)^\tran$ and we assume that
$Y_i$ are independent random variables with $\var(Y_i)=\sigma^2$.  The
design matrix is $\cx\in\real^{N\times 4}$ with $i$'th row
$(1, x_i, z_i, x_iz_i)$, and
$\var(\hat\beta)=(\cx^\tran\cx)^{-1}\sigma^2$.  Because $\sigma^2$
does not depend on $\Delta$, we can compare designs assuming that
$\sigma=1$.

Next, we look at how $\cx^\tran\cx$ depends on $\Delta$.
For large $N$ we can replace
$\sum_ix_i^2$ by $N\int_{-1}^1x^2\rd x/2=N/3$.
Similar integral approximations yield
\begin{align}\label{eq:xtx}
\frac1N\cx^\tran\cx \approx
\begin{pmatrix}
1 & 0 & 0 & \phi(\Delta)\\
0 & 1/3 & \phi(\Delta)& 0\\
0 & \phi(\Delta)  & 1 &0\\
\phi(\Delta) & 0 & 0 & 1/3
  \end{pmatrix},
\end{align}
where $\phi(\Delta)$ is the average value of $z\times x$ over
the design.  We let
$$
z(x) = \e(z\giv x) = \begin{cases}
  -1, & x \le -\Delta\\
  \phm0, & |x|<\Delta\\
  \phm1, & x \ge \Delta
\end{cases}
$$
and find that
\begin{align}\label{eq:defphi}
\phi(\Delta) = \frac12\int_{-1}^1
xz(x)\rd x
=\frac12\int_{-1}^{-\Delta}(-x)\rd x
+\frac12\int_{\Delta}^{1}x\rd x
=\frac{1-\Delta^2}2.
\end{align}

The approximation error in~\eqref{eq:xtx} is $O_p(1/\sqrt{N})$ when
the random $z_i$ are assigned by simple random sampling and it is much
smaller under stratified sampling. We will work with~\eqref{eq:xtx} as
if it were exact.

We can reorder the rows and columns of~\eqref{eq:xtx} to
make it block diagonal,
$$
\bordermatrix{& 1 & zx & z & x\cr
\,  1 & 1 & \phi & 0 & 0\cr
 zx & \phi & 1/3 & 0 & 0\cr
\, z & 0 & 0 & 1 & \phi\cr
\, x & 0 & 0 & \phi & 1/3\cr
}
$$
where the labels on the matrix above
refer to the variables that the $\beta_j$
multiply and $\phi=\phi(\Delta)$.
It follows that
\begin{align}\label{eq:varbetatwoline}
N\times\var\left(
\begin{pmatrix}
\hat\beta_0\\
\hat\beta_3\\
\hat\beta_2\\
\hat\beta_1
\end{pmatrix}
\right)
=\frac1{1/3-\phi^2}
\begin{pmatrix}
1/3 & -\phi & 0 & 0\\
-\phi &  1 & 0 & 0\\
0 & 0 & 1/3 & -\phi\\
0 & 0 & -\phi & 1\\
\end{pmatrix}.
\end{align}
%Thus the variances scale by $(1/3-\phi^2)^{-1}$.
The individual coefficients' variances are
$\var(\hat\beta_0)=\var(\hat\beta_2)=1/(1-3\phi^2)$
and
$\var(\hat\beta_1)=\var(\hat\beta_3)=3/(1-3\phi^2)$.
These variances are smallest for small values of $\phi$,
corresponding to large values of $\Delta$. That is, the
more randomized experimentation there is in the data, the less
variance there is in the estimates.
%From~\eqref{eq:defphi}, we see that $\phi$ is strictly
%decreasing from $1/2$ to $0$ as $\Delta$ increases
%from $0$ to $1$.
Therefore, the RDD
is least efficient and the RCT is most efficient.
Larger values of $\phi$ also induce stronger correlations
among the $\hat\beta_j$.

The estimated gain from the intervention
for a subject with a given $x$ is
$\hat\e(Y\giv x,z=1)-\hat\e(Y\giv x,z=-1)
=2(\hat\beta_2+x\hat\beta_3)$.
Next
\begin{align}\label{eq:varlingainatx}
\var(2(\hat\beta_2+x\hat\beta_3)) =
4\times\frac{1/3+x^2}{1/3-\phi^2}
=\frac{16(1+3x^2)}{1+3\Delta^2(2-\Delta^2)}
\end{align}
after some algebra.
The relative efficiency of the experiment versus regression
discontinuity is
\begin{align}\label{eq:eff}
\frac{\var\bigl(2(\hat\beta_2+x\hat\beta_3);\Delta=0\bigr)}
{\var\bigl(2(\hat\beta_2+x\hat\beta_3);\Delta=1\bigr)}
=\frac{1+3(2-1)}{1+3\times 0}=4
\end{align}
for all $x$.
That is, the randomized experiment with $N/4$ observations is
as informative as the regression discontinuity with $N$ observations
and this holds uniformly over all levels of the assignment variable $x$.
This is the factor of $4$ from \citet{Bloom08} mentioned earlier.

\begin{figure}[t]
\centering
\includegraphics[width=.9\hsize]{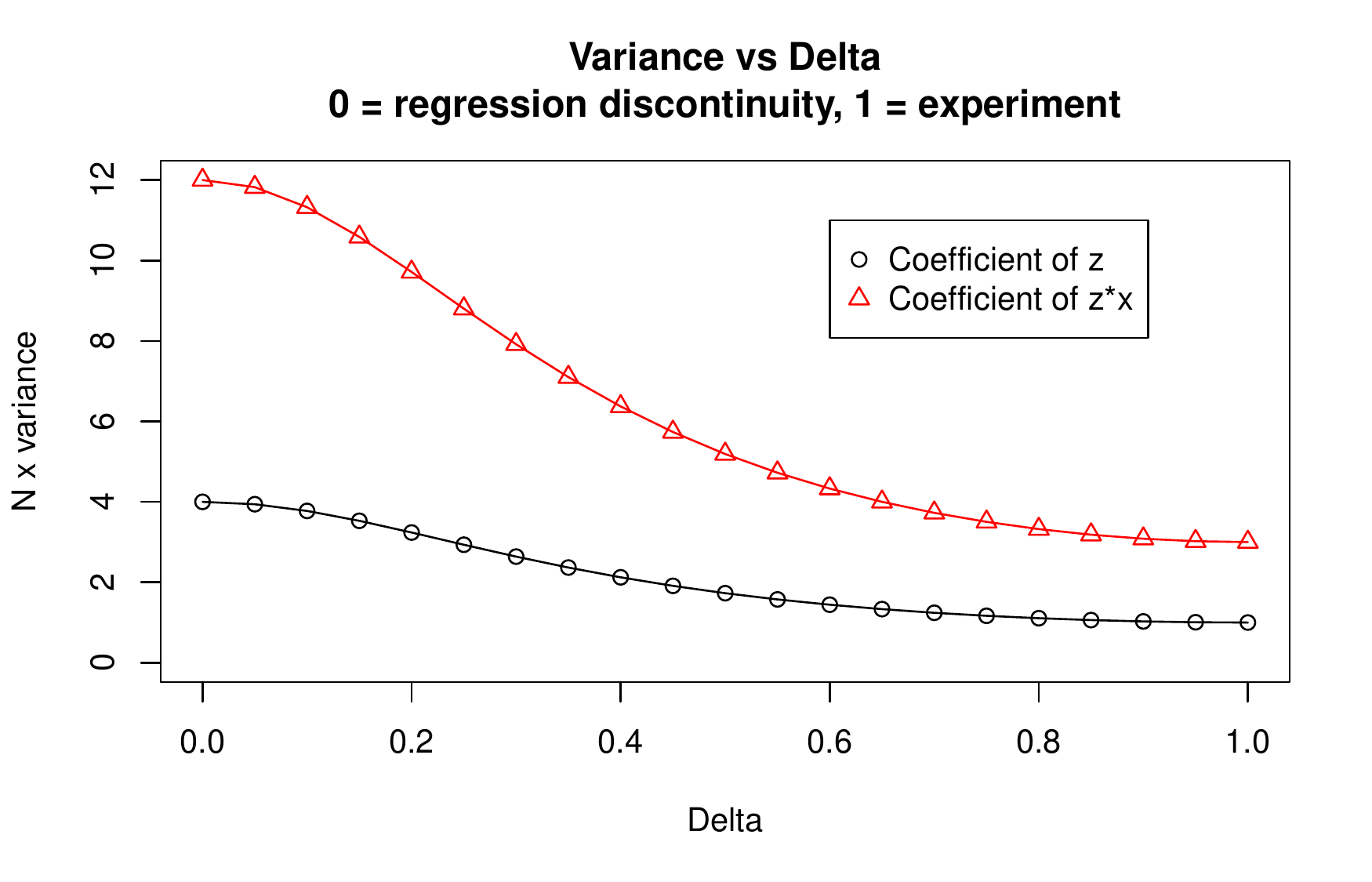}
\caption{\label{fig:zinfo}
  The top curve is the limiting value of $N\var(\hat\beta_3)$
  plotted versus the fraction $\Delta$ of experimental data
  in the hybrid.
  Here $\beta_3$ is the regression coefficient of $xz$.
  The bottom curve corresponds to $N\var(\hat\beta_2)$
  where $\beta_2$ is the coefficient of $z$.}
 \end{figure}

\begin{table}[b]
\centering
\begin{tabular}{llcc}
\toprule
Method & $\Delta$ & $\var(\hat\beta_2)$ & $\var(\hat\beta_3)$\\
\midrule
Regression discontinuity & $0$ & $4/N$ & $12/N$\\
Experiment &$1$ & $1/N$ & $3/N$\\
\bottomrule
\end{tabular}
\caption{\label{tab:variances}
Variance of $\hat\beta_2$ (treatment effect intercept)
and $\hat\beta_3$ (treatment effect slope)
under regression discontinuity ($\Delta=0$)
and randomized experiment ($\Delta=1$).
It assumes that $\var(Y\giv x,z)=1$.
}
\end{table}

Figure~\ref{fig:zinfo} shows the variance of the treatment effect
parameters as a function of $\Delta$.  Some values from the plot are
shown in Table~\ref{tab:variances}.  The regression
discontinuity design has four times the variance of the experiment as
we saw in equation~\eqref{eq:eff}.  The slope coefficient for treatment
always has three times the variance of the intercept coefficient as
follows from~\eqref{eq:varbetatwoline}.  Figure~\ref{fig:varslin} show
the variance of the estimated impact versus $x$ for several choices of
$\Delta$.

\begin{figure}[t]
\centering
\includegraphics[width=.9\hsize]{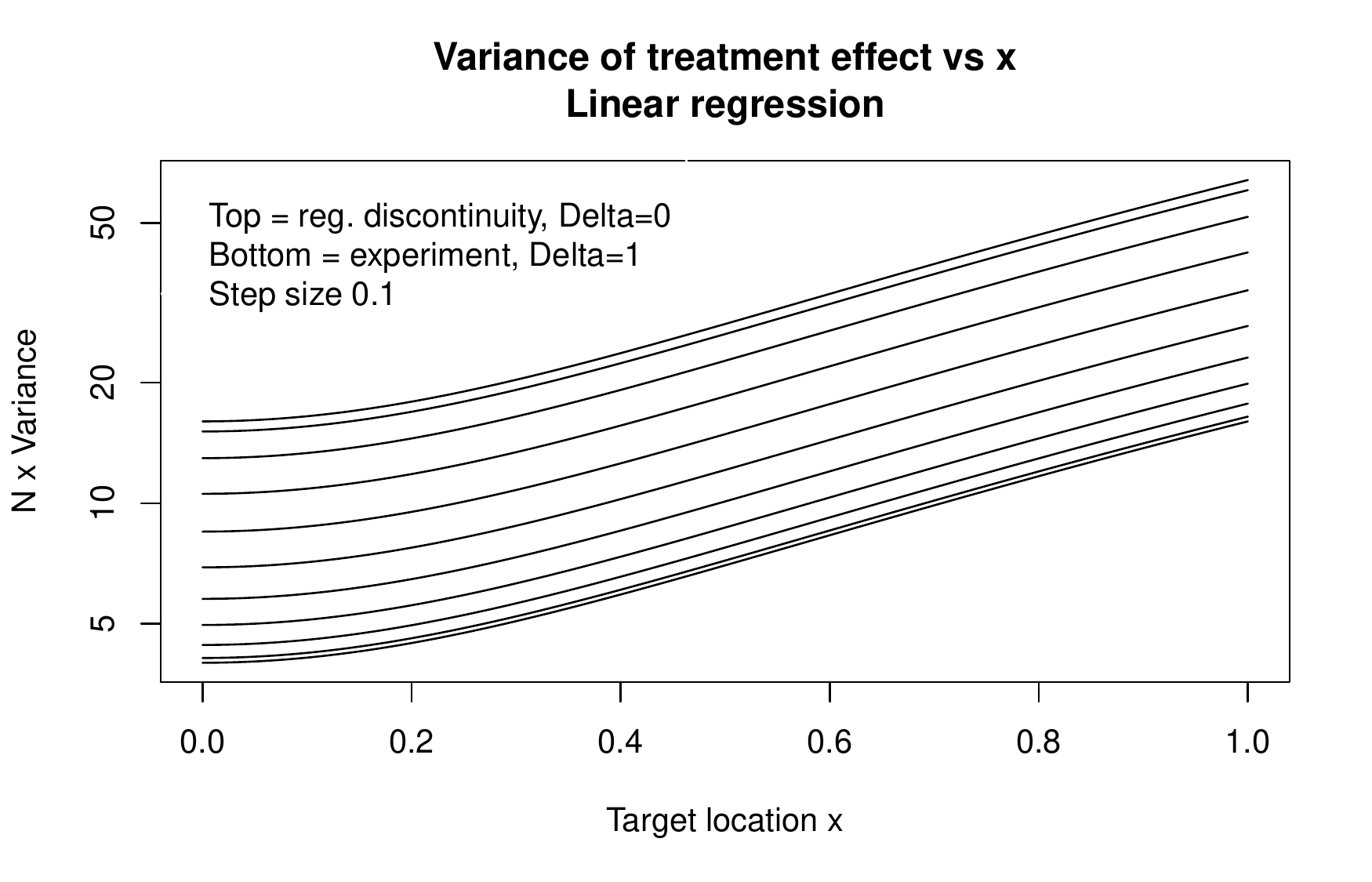}
\caption{\label{fig:varslin}
Variance of $2(\hat\beta_2+x\hat\beta_3)$
versus $x$  in the two-line
model~\eqref{eq:twolines},
for $\Delta$ between $0$ and $1$ in steps of $0.1$.
Note that the vertical axis is logarithmic.
}
 \end{figure}

\section{Cost of experimentation}\label{sec:cost}

We ordinarily expect the value of the treatment to increase
with the variable $x$.  In that case the greatest return on
the $N$ subjects in the experiment arises from the regression
discontinuity design with $\Delta=0$.
The information gain from $\Delta>0$ comes at some
cost in the present sample. This section quantifies
that cost.

For a deterministic allocation of $z=1$ or $z=-1$ we have
$
\e(Y\giv x,z) = \beta_0+\beta_1x+\beta_2z+\beta_3zx.
$
When $z$ is chosen randomly with $\Pr(z=1)=\Pr(z=-1)=1/2$, then
$
\e(Y\giv x) = \beta_0+\beta_1x.
$
It follows that the expected gain per subject in the hybrid design is
\begin{align*}
g(\Delta)=\e(Y)&\equiv\frac12\int_{-1}^{-\Delta}(\beta_0+\beta_1x-\beta_2-\beta_3x)\rd x
+\frac12\int_{-\Delta}^{\Delta}(\beta_0+\beta_1x)\rd x\\
&\quad+\frac12\int_{\Delta}^{1}\beta_0+\beta_1x+\beta_2+\beta_3x\rd x\\
&=\beta_0+\beta_3(1-\Delta^2)/2.
  \end{align*}

Neither $\beta_1$ nor $\beta_2$ appear in this gain and
the value of $\beta_0$ does not affect our choice of $\Delta$.
Only $\beta_3$ which models how the payoff from the incentive
varies with the assignment variable $x$ makes a difference.
Compared to the regression discontinuity design with $\Delta=0$,
the cost of incorporating experimentation is
$$
%N(g(0)-g(\Delta)) = N\beta_3\Delta^2/2,
N(g(0)-g(\Delta)) = \frac{N\beta_3\Delta^2}2,
$$
which grows slowly as $\Delta$ increases from zero
and then rapidly as $\Delta$ approaches one.
If $\beta_3>0$, then as expected, we gain the most
from the regression discontinuity design and the least from
the experiment.  This is a classic exploration-exploitation
tradeoff.

It is possible that some settings have $\beta_3<0$.
This might happen if the incentive is additional free tutoring
in the educational context and the strongest students don't need it,
or if it is advice on how to
best use an e-commerce company's products in a context where higher
performing customers already knew about the advice.
In these cases the greatest gain comes from giving the incentive to
the bottom $N/2$ customers and not the top $N/2$ customers.
The analysis of this paper goes through by reversing the
customer ranking, thereby replacing $x$ by $-x$ and
also changing the sign of $\beta_3$.

Now we turn to optimizing the choice of $\Delta$ given some
assumptions on the relative value of the information in the
data for future decisions and the expected gain on the
experiment. The precision (inverse variance) of our estimate of $\hat\beta$
is a linear function of $N$ and so is the expected gain.
We can therefore trade off precision per subject with gain per subject.
We think that $\beta_3$ is the most important parameter so we take
the precision gain per subject to be
\begin{align}\label{eq:precision}
\pre(\Delta) \equiv \frac1{N\var(\hat\beta_3)} = \frac13 -\phi^2 = \frac13 - \frac{(1-\Delta^2)^2}4.
\end{align}
Alternatively, we could focus on $2\beta_2$ which
is both the average gain per subject and the gain for a subject at $x=0$.
The precision for $2\beta_2$ turns out to be $(3/4)\pre(\Delta)$ so it
is perfectly aligned with precision on $\beta_3$.
More generally, the gain from the incentive at any specific $x$
has a variance given by~\eqref{eq:varlingainatx}. Any weighted
average of precision of $2(\beta_2+\beta_3x)$ over points $x\in[-1,1]$ is a scalar multiple of $\pre(\Delta)$
from~\eqref{eq:precision}.

We trade off gain per subject and precision per subject with
the value function
\begin{equation}\label{eq:value}
\begin{split}
v(\Delta) &= g(\Delta) + \lambda\cdot \pre(\Delta)
%\\&
=\beta_0 +\beta_3\frac{1-\Delta^2}2 + \lambda \Bigl( \frac13 - \frac{(1-\Delta^2)^2}4\Bigr),
\end{split}
\end{equation}
where $\lambda>0$ measures the value for future decisions of having greater precision
on $\beta_3$.  Because $\lambda$ is about information gain for the future
we consider it to have `long term' value while $\beta_3$ describes value in the
immediate data set, a relatively `short term' consideration.
\begin{proposition}\label{prop:bestdelta}
  Let $v(\Delta)$ be given by equation~\eqref{eq:value}
  with  $\lambda>0$ and $\beta_3\ge0$.   Then the maximum of $v$ over $\Delta\in[0,1]$
occurs at
  \begin{equation}\label{eq:deltopt}
\Delta_*=
\begin{cases}
1, &\phantom{0\le}\beta_3/\lambda\le0\\
\sqrt{1-\beta_3/\lambda}, & 0 \le \beta_3/\lambda\le 1\\
0, &1\le \beta_3/\lambda.\\
\end{cases}
  \end{equation}
\end{proposition}
\begin{proof}
Let $\gamma = \Delta^2$. We will first maximize $v= c -\beta_3\gamma/2 -\lambda(1-\gamma)^2/4$
over $0\le\gamma\le1$,
where $c$ does not depend on $\gamma$. Now $v$ has a unique maximum over $\gamma\in\real$
at $\gamma_* = 1-\beta_3/\lambda$. The maximizing $\gamma$ is
$\gamma_*$ when $0\le\gamma_*\le1$, it is $0$ when $\gamma_*<0$ and it is $1$
when $\gamma_*>1$.
Equation~\eqref{eq:deltopt} translates these results back to the optimal $\Delta$.
\end{proof}

We see from equation~\eqref{eq:deltopt} that the decision depends
on the critical ratio $\beta_3/\lambda$.  The numerator reflects
the value of more efficient allocation and the denominator captures
the value of improved information gathering.
When $\beta_3\ge\lambda$ then the RDD with $\Delta=0$
is optimal. The full experiment, $\Delta=1$, is never optimal unless
$\beta_3=0$ or the value $\lambda$ of information to be used in future
decisions is infinite.

\begin{figure}[t]
\centering
\includegraphics[width=.9\hsize]{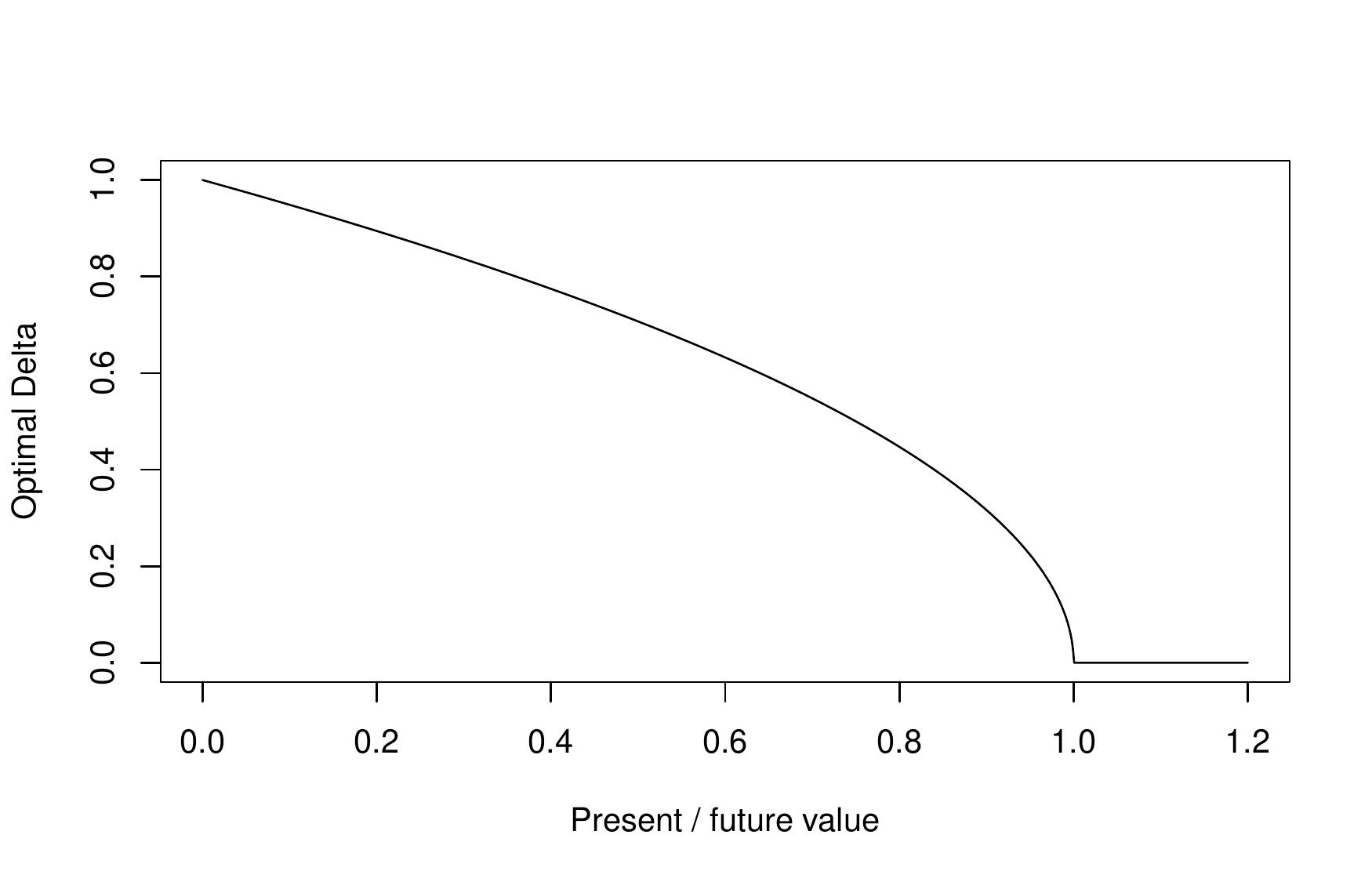}
\caption{\label{fig:optdel}
  The horizontal axis has $\beta_3/\lambda$ where
  $\beta_3$ affects immediate gain per subject of the
  treatment allocation and $\lambda$ quantifies the value
  of precise slope estimation. The vertical axis gives
  the optimal $\Delta$ in a tie-breaker experiment.
}
 \end{figure}

 Figure~\ref{fig:optdel} shows the value $\Delta_*$ from
 equation~\eqref{eq:deltopt} versus the ratio $r=\beta_3/\lambda$ of
 the short term to long term value coefficients.  The function is
 nearly equal to $1-r/2$ near the origin and has negative curvature on
 $0\le r\le1$. If future uses are important enough that $r\le1/10$,
 then one should use $\Delta\ge 1-0.1/2=0.95$.  That is, when the
 future is very important the optimal hybrid is very close to an RCT.

In practice it may well be difficult to
choose $\Delta$ to maximize the value~\eqref{eq:value}
because we don't know what $\lambda$ to
choose and because the tradeoff
depends on $\beta_3$ about which we may have little prior knowledge.
The parameter $\lambda$ will be hard to choose because
it quantifies the relative value of future
information versus the present value of the intervention.
A practical approach is to use the smallest
experiment with at least some given proportion $\rho$ of the information
available from the RCT.  That is, for some $\rho\in[1/4,1]$, choose the
smallest $\Delta$ with $\pre(\Delta)\ge \rho\times \pre(1)$.
We don't need to consider $\rho<1/4$ because even $\Delta=0$ has
at least one fourth the efficiency of the RCT.

\section{Quadratic regression}\label{sec:quadratic}

A quadratic regression model of the form
\begin{align}\label{eq:twocurves}
\e(Y) = \beta_0+\beta_1x+\beta_2z+\beta_3xz+\beta_4x^2+\beta_5x^2z
  \end{align}
allows a richer exploration of the treatment effect.
For instance, model~\eqref{eq:twocurves}
allows for the possibility that the treatment pays off if and only
if $x$ is in some interval.  It also allows for a situation where
the payoff only comes outside of some interval.
This model has even (symmetric) predictors
$1$, $xz$, $x^2$ and odd (antisymmetric) predictors
$x$, $z$, $zx^2$.  As in the linear case, the even and odd
predictors are orthogonal to each other.

Now $(1/N)\cx^\tran\cx$ is a $6\times 6$ block diagonal matrix.
Some of the entries are
$$\phi_3 \equiv\phi_3(\Delta) = \frac12\int_{-1}^1\e(z\giv x)x^3\rd x=\int_\Delta^1 x^3\rd x =
\frac{1-\Delta^4}4$$
as well as $\phi(\Delta)$ from Section~\ref{sec:efftwo}
that  we call $\phi_1(\Delta)$ here.
We find that
\begin{align}\label{eq:xtxbyn}
\frac1N\cx^\tran\cx =
\bordermatrix{& 1 & zx & x^2 & z & x & zx^2\cr
\,  1 & 1 & \phi_1 & 1/3 & 0 & 0 & 0\cr
 zx & \phi_1 & 1/3 & \phi_3 & 0 & 0&0\cr
\, x^2 & 1/3 & \phi_3 & 1/5 & 0 & 0&0\cr
\, z & 0 & 0 & 0 & 1 & \phi_1 & 1/3 \cr
\, x & 0 & 0 & 0 & \phi_1 & 1/3& \phi_3\cr
\, zx^2 & 0 & 0 & 0 & 1/3 & \phi_3 & 1/5\cr
}
\end{align}
after ignoring sampling or stratified sampling fluctuations.
Once again we get a block diagonal pattern with two identical blocks.
This is a consequence of $z^2=1$, and it will happen
for more general models with odd and even predictors.

\begin{proposition}\label{prop:quadstuff}
For $N>0$, let $\cx^\tran\cx$ be given by~\eqref{eq:xtxbyn}.
Then
\begin{align}\label{eq:xtxinvquad}
(\cx^\tran\cx)^{-1}
=\frac1{ND(\Delta)}
\begin{pmatrix}
M(\Delta) & 0\\
0 & M(\Delta)  \end{pmatrix}
\end{align}
for a $3\times 3$ symmetric matrix
$$
M(\Delta) = \begin{pmatrix}
\dfrac1{15}-\phi_3^2 & \dfrac{\phi_3}3-\dfrac{\phi_1}5 & \phi_3\phi_1-\dfrac19\\[2ex]
'' & \dfrac4{45} & \dfrac{\phi_1}3-\phi_3\\[2ex]
'' & '' & \dfrac13-\phi_1^2
\end{pmatrix},
$$
and a determinant
$
D(\Delta) =
4/{135}-{\phi_1^2}/5-\phi_3^2+(2/3)\phi_1\phi_3.
$
\end{proposition}
\begin{proof}
  Multiplying $M(\Delta)$ above by the upper left $3\times3$ submatrix in~\eqref{eq:xtxbyn}
  yields $I_3$ times $D(\Delta)$, after some lengthy manipulations.
\end{proof}

Figure~\ref{fig:vars} show the variance of the estimated impact versus
$x$ for several choices of $\Delta$.  Notice that the variance is
given on a logarithmic scale there.  The regression discontinuity
design $\Delta=0$ in the top curve there, has extremely large
variances especially where $|x|$ is close to $1$. The randomized
design at the bottom has much smaller variance.  Even the maximum
variance in the RCT (at $x=1$) is smaller than the minimum variance in
the RDD (at $x=0$).

\begin{figure}
\centering
\includegraphics[width=.9\hsize]{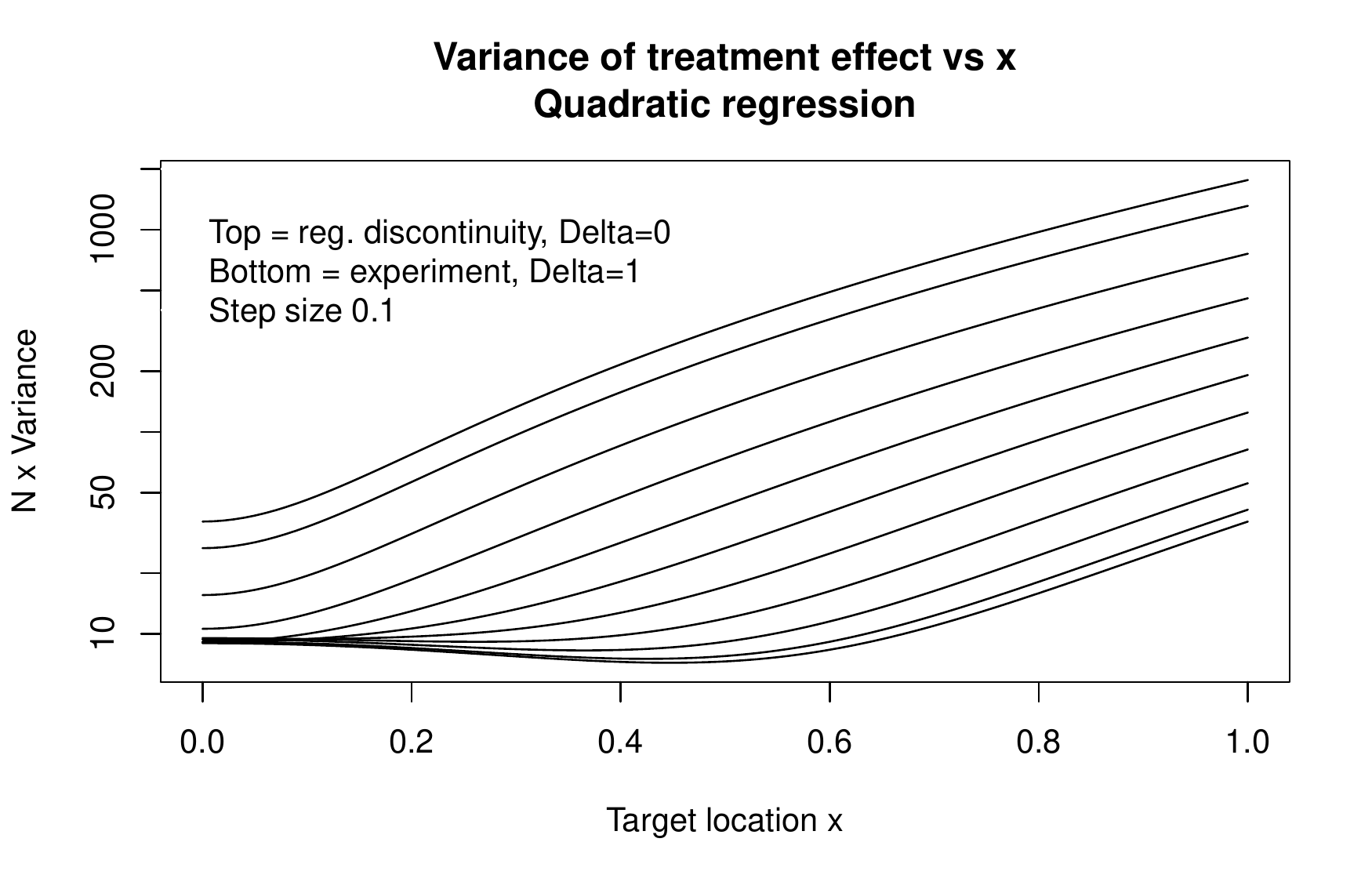}
\caption{\label{fig:vars}
Variance of $2(\hat\beta_2+x\hat\beta_3+x^2\hat\beta_5)$
versus $x$ in the quadratic model~\eqref{eq:twocurves},
for $\Delta$ between $0$ and $1$ in steps of $0.1$.}
\end{figure}

\section{Gaussian case}\label{sec:gaussian}

In some settings, the original assignment variable $x$ might have a
nearly Gaussian distribution.  By changing location and scale we can
suppose that $x$ has approximately the $\dnorm(0,1)$ distribution,
without loss of generality.  We use $\varphi(\cdot)$ and $\Phi(\cdot)$
to represent the $\dnorm(0,1)$ probability density function and
cumulative distribution function, respectively.

We will experiment on the central data
with $|x_i|\le\tau$ choosing $\tau$ to get
a fraction $\Delta$ of data in the experiment.
That leads to $\tau = \Phi^{-1}( (1+\Delta)/2)$.
After reordering the variables we find in this case that
$$
\frac1N\cx^\tran\cx \approx
\bordermatrix{& 1 & zx & z & x\cr
\,  1 & 1 & \phi_G & 0 & 0\cr
 zx & \phi_G & 1 & 0 & 0\cr
\, z & 0 & 0 & 1 & \phi_G\cr
\, x & 0 & 0 & \phi_G & 1\cr
}.
$$
The value of $\phi$ from the uniform case changes to
\begin{align*}
\phi_G & =
\int_{-\infty}^{-\tau}(-x)\varphi(x)\rd x
+\int_{\tau}^{\infty}x\varphi(x)\rd x
  =2\int_{\tau}^{\infty}x\varphi(x)\rd x\\
&=2\varphi(\tau) = 2\varphi( \Phi^{-1}((1+\Delta)/2)).
\end{align*}
Compared to the uniform scores case, the diagonal has changed from $(1,1/3,1,1/3)$ to $(1,1,1,1)$.
Now
\begin{align}\label{eq:varbetatwolinegaus}
N\times\var\left(
\begin{pmatrix}
\hat\beta_0\\
\hat\beta_3\\
\hat\beta_2\\
\hat\beta_1
\end{pmatrix}
\right)
=\frac1{1-\phi_G^2}
\begin{pmatrix}
1 & -\phi_G & 0 & 0\\
-\phi_G &  1 & 0 & 0\\
0 & 0 & 1 & -\phi_G\\
0 & 0 & -\phi_G & 1\\
\end{pmatrix}.
\end{align}
For this Gaussian case, all $4$ estimated coefficients $\hat\beta_j$
have the same variance, equal to $1/(1-\phi_G^2)$.  The variances for
uniform assignment variables were not all the same.  The difference stems
from the points $x_i$ having variance $1/3$ in the uniform case
instead of variance $1$ here.  As before as $\Delta$ increases,
$\phi_G$ also increases and so $\var(\hat\beta_j)$ decreases.

Now we work out the efficiency of the RCT compared to the RDD.  For
the RCT, $\Delta=1$ yields $\tau =\infty$ and then $\phi_G=0$. For the
RDD, $\Delta=0$ yields $\tau=0$ and then $\phi_G=2\varphi(0)$.  Thus
the efficiency of the RCT compared to the RDD is
\begin{align*}
\frac{1}{ 1-[2\varphi(0)]^2}
&= \frac{\pi}{\pi-2}\doteq 2.75
\end{align*}
as reported by \citet{gold:1972}.
This is somewhat less than the efficiency gain of $4$ in
the uniform case.
The efficiency versus $\Delta$ (not shown) has a
qualitatively similar shape to the black curve for the coefficient of $z$
in the uniform case (Figure~\ref{fig:zinfo}).

\section{Sliding scales}\label{sec:sliding}

In the tie-breaker design, there are three levels of subjects
getting the treatment condition with probabilities 0\%, 50\% and 100\%.
We could use a more general sliding scale where this probability
rose from $0$ to $100$\% in a sequence of smaller steps, or even
rose continuously as a function $p(x)$ of the assignment variable $x\in[-1,1]$.
Figure~\ref{fig:carpentry} has an example of each type.
We show here that there is little to gain from such a sliding scale
in the case where half the subjects will be treated
and half will not.
At the end of this section we point to an advantage
of using treatment probabilities
$\epsilon$, $1/2$ and $1-\epsilon$ where $0<\epsilon<1/2$.

\begin{figure}
\centering
\includegraphics[width=.9\hsize]{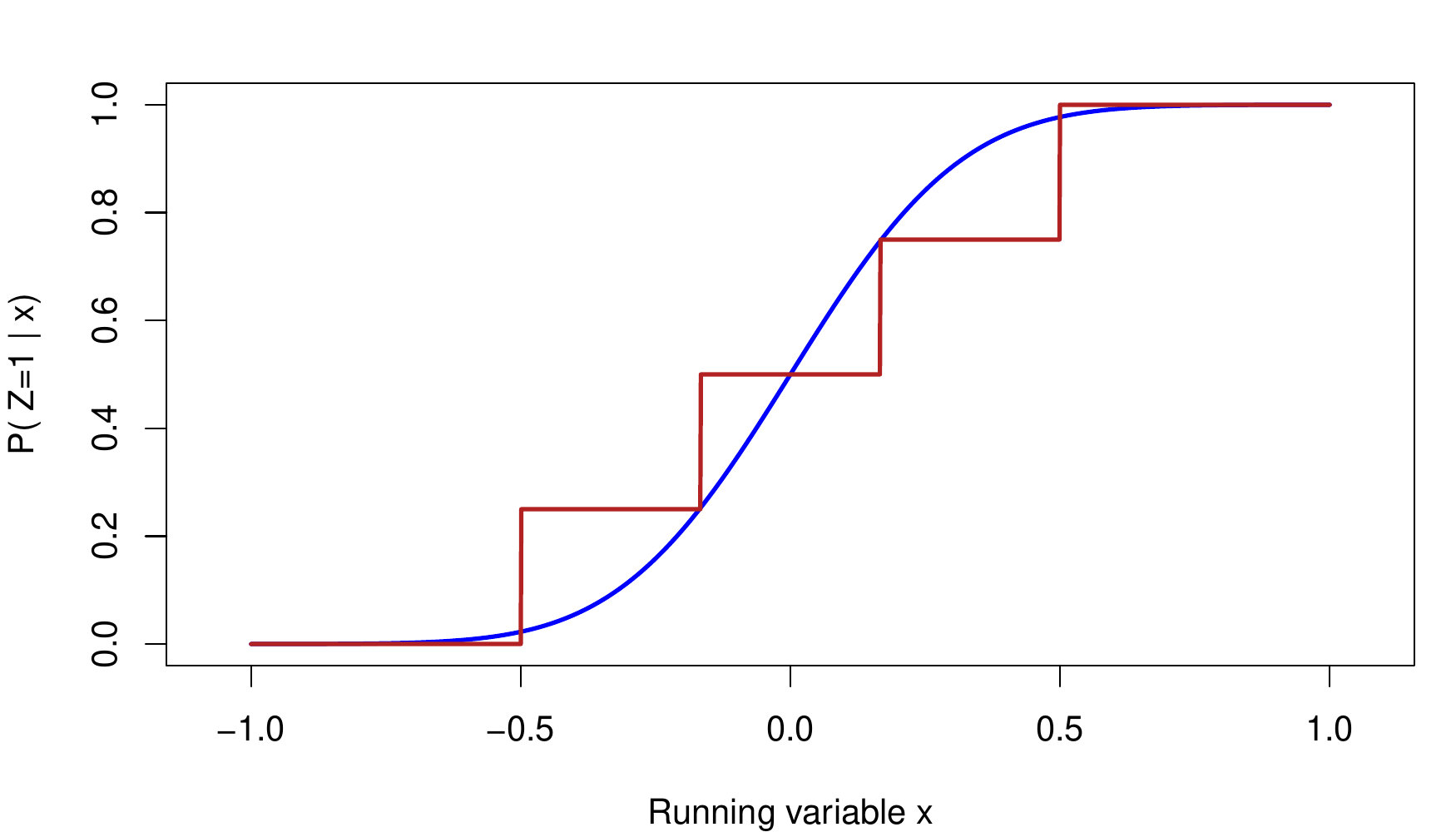}
\caption{\label{fig:carpentry}
The smooth curve shows a sliding
scale with $\Pr(z=1\giv x)=\Phi(4x)$,
where $\Phi$ is the $\dnorm(0,1)$ cumulative distribution function.
The step function has $\Pr(z=1\giv x)$ increasing from $0$ in steps.
}
\end{figure}

Suppose first that $p(x)$ satisfies $p(x)=1-p(-x)$
and is non-decreasing. For instance, $p$ could be the
cumulative distribution function (CDF) of a symmetric distribution.  A proper CDF would
ordinarily have $p(-1)=0$ and $p(1)=1$ too,
but we do not need to impose that.
With this $p(x)$, the expected number of treated
cases is $N/2$.
The variance-covariance matrix of $\hat\beta$ in our model
is then the same as in~\eqref{eq:xtx} with
$\phi(\Delta)$ replaced by
$$\overline{zx}\equiv
\frac12\int_{-1}^1 x\e(z\giv x)\rd x
=\frac12\int_{-1}^1 x(2p(x)-1)\rd x.
$$
For symmetric $p$,
$\overline{z} = (1/2)\int_{-1}^1 (2p(x)-1)\rd x=0$
and
$\overline{zx^2} = (1/2)\int_{-1}^1 x^2(2p(x)-1)\rd x=0$.

Now let's consider the expected gain on the $N$ subjects in the trial.
The short term gain, averaged over $z$, is
$$
\frac12\int_{-1}^1 \beta_0 + \beta_1x+\beta_2\e(z\giv x)\rd x =
\beta_0+\beta_3\overline{zx}.
$$
Supposing as before that $\beta_3>0$, we see that the tradeoff between
immediate value and information gained is driven by the single
variable $\overline{zx}$ and not by whether $\e(z\giv x)$ is on a
continuous sliding scale or simply at three levels $0$, $0.5$ and
$1.0$.  For tie-breaker designs $\overline{zx} =
(1-\Delta^2)/2\in[0,1/2]$ for $0\le\Delta\le 1$.  If $p(x)$ is
symmetric then we find that
$$
\overline{zx} = 2\int_0^1xp(x)\rd x-\frac12.
$$
For $x>0$, we know that $1/2 = p(0)\le p(x)\le1$.  It follows that
$0\le\overline{zx}\le1/2$. In other words, the full range of
exploration-exploitation tradeoffs available from a sliding scale with
symmetric $p(x)$ is already available in a tie-breaker design.

Now suppose that we relax the symmetry
constraint on $p(x)$, while still having 50\% allocation, that is $\bar z=0$.
The short term gain still depends only on $\overline{zx}$.
The other ingredient in the tradeoff is
$\var(\hat\beta)$ which depends on $\bar z$,
$\overline{zx}$ and $\overline{zx^2}$.
To see what can be gained by symmetrizing we compare $p(x)$
to a symmetric alternative $\tilde p(x) = (p(x)+1-p(-x))/2$.
This has $\tilde p(x) +\tilde p(-x) = 1$, and so $\tilde p$
is symmetric as described above.
Denote the result of replacing $p$ by $\tilde p$ in
the definitions of $\bar z$, $\overline{zx}$ and
$\overline{zx^2}$ by $\tilde z$, $\wt{zx}$ and $\wt{zx^2}$, respectively.
We find that $\tilde z =\bar z$ and $\overline{zx}=\wt{zx}$,
and so symmetrizing has not changed these quantities.
However, symmetrizing makes $\wt{zx^2}=0$ which is not necessarily equal to $\overline{zx^2}$.
We will see that $\overline{zx^2}$ enters our
expression for $\var(\hat\beta)$
only through $\overline{zx^2}^2\ge\wt{zx^2}^2=0$.

After some algebra our expression for $N\var(\hat\beta)$ yields
$$
N\var\begin{pmatrix}\hat\beta_0\\ \hat\beta_1\end{pmatrix}
\doteq N\var\begin{pmatrix}\hat\beta_2\\ \hat\beta_3\end{pmatrix}
\doteq
(A-BA^{-1}B)^{-1}
$$
where
$$
A=
\begin{pmatrix}
1 & 0\\
0 & \frac13
\end{pmatrix}\quad\text{and}\quad
B=
\begin{pmatrix}
\barz^2 +3\barzx^2& \barz\,\barzx +3\barzx\,\barzxx \\
\barz\,\barzx +3\barzx\,\barzxx
& \barzx^2+3\barzxx^2
\end{pmatrix}.
$$
Then
\begin{align}\label{eq:varsym}
\begin{split}
\var(\hat\beta_1)&\doteq\var(\hat\beta_3) \doteq\frac1N\frac{1-3\barzx^2}D,\quad\text{and} \\
\var(\hat\beta_0)&\doteq\var(\hat\beta_2) \doteq\frac1N\frac{1/3-\barzx^2-3\barzxx^2}D,
\end{split}
\end{align}
for a determinant
$$
D = \frac13 - 2\overline{zx}^2-3\overline{zx^2}^2+3\overline{zx}^4.
$$

Symmetrizing can increase but not decrease $D$ because $\wt{zx^2}^2=0$.
Symmetrizing does not change the numerators
for $\hat\beta_1$ and $\hat\beta_3$ and so it
can reduce but not increase their approximate variance expressions.
The cases of $\hat\beta_0$ and $\hat\beta_2$ are more
complicated because symmetrizing changes  both their
numerators and denominators.
However some straightforward calculus shows that
those expressions are minimized when $\overline{zx^2}^2=0$
so they cannot be increased by symmetrization.

It is possible that $\var(c^\tran\hat\beta)$
can be increased by symmetrization for some
values of $c\in\real^4$.  An example of this type
can be constructed with $p(x)=|x|$ and
$\tp(x)=1/2$.  Then $\var( \hat\beta_1+\hat\beta_3)$
is increased under symmetrization. The same
holds for $\var(\hat\beta_2+\hat\beta_3)$ where
$\hat\beta_2+\hat\beta_3$
is half of the expected treatment gain at $x=1$.

The most consequential coefficient is $\beta_3$.
A symmetric sliding scale cannot improve its
estimation compared to a tie-breaker design
at a given short term cost.
A non-symmetric sliding scale cannot improve
over a symmetric one when half of the cases are treated.
Thus, when half of the cases are to be treated,
the original tie-breaker design is optimal at
any given level of $\overline{zx}$.

One drawback of using
treatment probabilities $0$ and $0.5$ and $1$
is that some of the potential treatment allocations are deterministic.
Methods based on the potential outcomes framework
(see \cite{imbens2015causal}) cannot then be readily applied.
We could instead use three levels $\epsilon$, $0.5$
and $1-\epsilon$ with the central $\Delta$ of subjects
having $p(x)=0.5$.  Then we find that the critical quantity
governing both statistical and allocation efficiency becomes
$$\barzx = \frac12\int_{-1}^{-\Delta}x(2\epsilon-1)\rd x
+\frac12\int_\Delta^1x(2(1-\epsilon)-1)\rd x
=(1-\epsilon)\frac{1-\Delta^2}2.
$$
The consequence is that we can find a design only
in the range $0\le \barzx\le (1-\epsilon)/2$
instead of $0\le \barzx\le 1/2$.
For small $\epsilon$, this is only a mild reduction in the attainable range,
and it still requires only three levels of treatment probability.

\section{General numerical approach}\label{sec:general}

The two line model for an assignment variable $x$ with
a symmetric distribution made it simple to study
central experimental windows of the form $(-\Delta,\Delta)$.
In that setting the means of $x_i$
and $z_i$ were both zero, and the variance of parameter
estimates depended simply on just one quantity $\Delta$.
We may want to use a more general regression model,
allow experimental windows that are not centered
around the middle value of $x$, have $x$ values
that are not uniform or Gaussian, and we might also want to use
models other than two regression lines.

There might even be more than one assignment variable as in
\cite{abdulkadiroglu2017impact}.  The price for this flexibility is
high; users have to answer some hard questions about their goals, and
then do numerical optimization over parameters with a potentially
expensive Monte Carlo inner loop.  In this section we show that the
inner loop can be done algebraically.  We also find that the full
experiment $\Delta=1$ with $\Pr(z_i=1)=1/2$ is variance optimal.

We suppose that prior to treatment assignment, subject $i$
has a known feature vector $F_i\in\real^d$ which includes
an intercept variable equal to $1$, but not the treatment
variable $z_i$. For instance in the linear and quadratic models,
the features $F_i$ are $(1,x_i)^\tran$ and $(1,x_i,x_i^2)^\tran$, respectively.
In the regression model
$$
Y_i = F_i^\tran\beta + z_i F_i^\tran\gamma +\err_i,
$$
we have $\e(Y_i) = F_i^\tran(\beta+\gamma)$ for the treated
subjects $i$ and $\e(Y_i) = F_i^\tran(\beta-\gamma)$ for the
others. Here $\gamma\in\real^d$ models the effect of treatment.

The generalized tie-breaker study
works with a vector $\theta\in\real^d$ and sets
$$
\Pr(z_i=1\giv x_i)= \begin{cases}
\phm1, &   \theta^\tran F_i \ge\Delta\\
\phm p, & |\theta^\tran F_i|<\Delta\\
-1, &   \theta^\tran F_i \le -\Delta,
\end{cases}
$$
for some fixed $p\in(0,1)$, not necessarily $1/2$.
Because $F_i$ contains an intercept term, the experimental
window $|\theta^\tran F_i|<\Delta$ need not be centered
on a central value of $\theta^\tran F_i$.
The analyst must now choose $\Delta\ge0$,
$\theta\in\real^d$ and $p\in(0,1)$.

The analogue of our previous approach is to find the matrix
$(\cx^\tran\cx)^{-1}$ where
$$
\cx^\tran\cx =
\begin{pmatrix}
  A & B \\
  B & A
  \end{pmatrix},\quad A = \sum_i F_iF_i^\tran,\quad B = \sum_i w_{i}F_iF_i^\tran,
%\quad D = \sum_i w_{i2}F_iF_i^\tran
$$
for
$$
w_{i}= \e( z_i \giv F_i)=
\begin{cases}
1, & \theta^\tran F_i \ge\Delta,\\
2p-1, & |\theta^\tran F_i|<\Delta,\\
-1, & \theta^\tran F_i \le -\Delta.
\end{cases}
$$
The lower right corner of $\cx^\tran\cx$
is $A$ because $\e(z_i^2\giv F_i)=1$.
Averaging over the outcomes of $z_i$ this way is
statistically reasonable when $N\gg d$.
If $\err_i$ are independent with mean zero and variance $\sigma^2$, then
$$
\var\left(
\begin{pmatrix}
  \hat\beta\\
  \hat\gamma
  \end{pmatrix}
\right)
=(\cx^\tran\cx)^{-1}\sigma^2.
$$
This averages over the outcomes $\err_i$ so that
they do not have to be simulated.

One can now do brute force numerical search
for good values of $\theta$ and $p$ and $\Delta$.
A good choice would yield a favorably small
$\var(\hat\gamma)$.  A bad choice will yield a larger
variance covariance matrix. A very bad choice would
lead to singular $\cx^\tran\cx$ and one would
of course reject the corresponding triple $(\theta,\Delta,p)$.
For instance, such a singularity
would happen if $\max_i \theta^\tran F_i <-\Delta$
which is an obviously poor choice because then no
subjects would be in the treatment group.

Using a formula for the inverse of a block matrix we get
$$
\var(\hat\gamma) = \var(\hat\beta) = (A-BA^{-1}B)^{-1}\sigma^2
$$
and $\cov(\hat\beta,\hat\gamma) = -A^{-1}B(A-BA^{-1}B)^{-1}\sigma^2$.
In an RCT with $p=1/2$ we have $B=0$. For $\Delta<\max_i(|\theta^\tran F_i|)$
certain components of $B$  become nonzero.
That can increase $BA^{-1}B$ but not decrease it. As a consequence,  $\var(\hat\gamma)$
cannot be made smaller than it is under the RCT for any choice
of $\Delta$ given $\theta$, when $p=1/2$.

\section{Non-central experimental regions}\label{sec:noncentral}

Our treatment of the two line model assumed that the
experimental region was in the center of the range
of the assignment variable.
A customer loyalty program might well reward
just the top few customers and a scholarship program
will ordinarily award scholarships to fewer than half of the students.
We analyze that case
and compare statistical efficiency of a tie-breaker
design in the upper quantiles to an RDD there.
We find that the tie-breaker experimenting on
the second decile is about $1.62$ times as
efficient as an RDD with a threshold at the 85'th percentile,
both of which offer the treatment to 15\% of subjects.
If we find upon seeing the data that the linear model
is too biased and reduce the bias by only looking
at the top $30$\% of data, then on that data subset,
using the tie-breaker becomes $1.63$\% times as efficient
as the RDD. That is, the efficiency is virtually the same.

To handle designs where fewer than half of the subjects
are treated we let
\begin{align}\label{eq:hybridzab}
\Pr(z_i=1\giv x_i) =
\begin{cases}
\phm1, & x_i \ge b\\
\phm p, & a<x_i<b\\
-1, & x_i \le a
\end{cases}
\end{align}
for $a\le b$ and $0<p<1$.
We abuse notation a little by having the function $p(x)$
subsume all three cases in~\eqref{eq:hybridzab}.
For a less expensive treatment
we might want to offer it to the
top $50$\% of subjects and then randomize it to the
bottom $50$\% and~\eqref{eq:hybridzab}
can handle this choice too.

Let the assignment variable $x\in\real$ be random with $\e(x^4)<\infty$.
Then letting $\cx$ be the design matrix in the two line regression,
and noting that $z^2=1$, we have
$$
\frac1N\cx^\tran\cx =
\begin{pmatrix}
1 & \e(x) & \e(z) & \e(xz)\\
\e(x) & \e(x^2) & \e(xz) &\e(x^2z)\\
\e(z) & \e(xz) & 1 & \e(x)\\
\e(xz) & \e(x^2z) & \e(x) & \e(x^2)\\
  \end{pmatrix}
+O_p\Bigl(\frac1{\sqrt{N}}\Bigr)
$$
under random sampling of $x_i$ and $z_i$ given $x_i$ for $i=1,\dots,N$.
The $O_p(N^{-1/2})$ error holds because $\e(x^4)<\infty$.
The error could be $o_p(N^{-1/2})$ if $p(x)$ is a simple enough
function to make stratification tractable.

We can center $x$ so that $\e(x)=0$ and then
$$
\var(\hat\beta) \doteq
\frac1N\begin{pmatrix}
  D & C \\
  C & D
\end{pmatrix}^{-1},\ \text{for}\ \,
C =
\begin{pmatrix}
\e(z) & \e(xz)\\
\e(xz) &\e(x^2z)\\
\end{pmatrix}\ \,\text{and}\ \,
D =
\begin{pmatrix}
1 & 0\\
0 &\e(x^2)\\
\end{pmatrix}.
$$
We can scale $x$ to get $\e(x^2)=1$ so that $D=I_2$.
We retain more general scaling because $x\sim \dunif[-1,1]$
has $\e(x^2) =1/3$ and rescaling would require
working with the less convenient distribution
$\dunif[-\sqrt{3},\sqrt{3}]$.

We need the inverse of a block diagonal matrix containing
just two unique square blocks.  The following proposition specializes
block matrix inversion to our case.

\begin{proposition}\label{prop:2x2}
  Let $D$ be an invertible matrix and $C$
be a square matrix with the same dimensions as $D$.
If $D-CD^{-1}C$ is invertible, then
  $$
  \begin{pmatrix}
    D & C\\
    C & D
    \end{pmatrix}^{-1}
=
  \begin{pmatrix}
    A & B\\
    B & A
    \end{pmatrix}
  $$
for $A=(D-CD^{-1}C)^{-1}$ and $B=-ACD^{-1}$.
  \end{proposition}
\begin{proof}
Multiplying,
$$
  \begin{pmatrix}
    A & B\\
    B & A
    \end{pmatrix}
  \begin{pmatrix}
    D & C\\
    C & D
    \end{pmatrix}
=  \begin{pmatrix}
    AD+BC & AC+BD\\
    BD+AC & BC+AD
    \end{pmatrix}.
$$
Now $AC+BD=AC-ACD^{-1}D=0$
and $AD+BC=A(D-CD^{-1}C)=I$.
  \end{proof}

Using Proposition~\ref{prop:2x2} we get
$$
\var(\hat\beta) \doteq\frac1N
\begin{pmatrix}
\phm(D-CD^{-1}C)^{-1}\phantom{CD^{-1}} & -(D-CD^{-1}C)^{-1}CD^{-1} \\
-(D-CD^{-1}C)^{-1}CD^{-1}& \phm(D-CD^{-1}C)^{-1}\phantom{CD^{-1}}  \\
  \end{pmatrix}.
$$
Our primary interest is in $\var(\hat\beta_3)$,
for the coefficient of $xz$.  This is the
lower right element of $(D-CD^{-1}C)^{-1}$.
Now
\begin{align*}
D-CD^{-1}C
&= \begin{pmatrix}
1-\e(z)^2-\e(xz)^2/\e(x^2)
 &-\e(xz)\e(z)-\e(x^2z)\e(xz)/\e(x^2)\\
-\e(xz)\e(z)-\e(x^2z)\e(xz)/\e(x^2)
 &\e(x^2)-\e(xz)^2-\e(x^2z)^2/\e(x^2)\\
\end{pmatrix}\\
&\equiv
\begin{pmatrix} M_{11} & M_{12}\\
  M_{12} & M_{22}
  \end{pmatrix},
\end{align*}
and so
$$
\var(\hat\beta_3) = \frac1N\frac{M_{11}}{M_{11}M_{22}-M_{12}^2}.
$$

The asymptotic value of $N\var(\hat\beta_3)$ depends on certain integrals.
For the case of primary interest to us with $x\sim \dunif[-1,1]$,
and $p(x)=1/2$ in the experimental region, these are
\begin{align*}
\e(x^2) &= \frac12\int_{-1}^1 x^2\rd x = \frac13,\\
\e(xz) & = \frac12\int_{-1}^a(-x)\rd x + \frac12\int_{b}^1x\rd x
%=\frac{1-a^2}4+\frac{1-b^2}4
=\frac12-\frac{a^2+b^2}4,\\
 \e(z) & = -\frac12(a+1) +\frac12(1-b) = -\frac{a+b}2,\quad\text{and}\\
\e(x^2z) & = \frac12\int_{-1}^a(-x^2)\rd x + \frac12\int_{b}^1x^2\rd x
%=\frac{-1-a^3}6+\frac{1-b^3}6
=-\frac{a^3+b^3}6.
  \end{align*}

\begin{table}\centering
\begin{tabular}{lrrr}
\toprule
Method & $a\ \,$ & $b\ \,$ & $\var(\hat\beta_3)$\\
\midrule
Experiment &$-1.00$ & $1.00$ & $3.00/N$\\
RDD & $0.00$ & $0.00$ & $12.00/N$\\
Bottom 50\% & $-1.00$ & $0.00$ & $13.09/N$\\
\midrule
Skew RDD (85th) &  $0.70$ & $0.70$ & $223.44/N$\\
Second 10\% & $0.60$ & $0.80$ & $137.56/N$\\
%Middle 10\% & $-0.10$ & $0.10$ & $11.32/N$\\
%Skew RDD (90th) &  $0.80$ & $0.80$ & $751.03/N$\\
%Skew RDD (80th) &  $0.60$ & $0.60$ & $95.21/N$\\
%Tiny (median) &  $-0.01$ & $0.01$ & $11.99/N$\\
%Tiny (90th) &  $0.79$ & $0.81$ & $739.96/N$\\
%Tiny (80th) &  $0.59$ & $0.61$ & $94.86/N$\\
\bottomrule
\end{tabular}
\caption{\label{tab:varoffctr}
Variance of $\hat\beta_3$ (treatment effect slope)
for some central and non-central experimental regions.
}
\end{table}

Table~\ref{tab:varoffctr} shows $\var(\hat\beta_3)$
for various designs when $x\sim\dunif[-1,1]$. The first two are the
full experiment and the RDD discussed previously.
Next is an experiment on just the bottom half of $x$.
This strategy is inadmissible by our criteria. It has
more variance than the RDD and also lower allocation
efficiency.

Next, the table compares some options we might have
when only $15$\% of subjects can get the treatment.
The first one is to do an RDD with the critical point
at the $85$'th percentile.
Alternatively we could choose a tie-breaker design
giving the top $10$\% of subjects the treatment
along with a randomly chosen half of the second $10$\%
of customers.
The skewed RDD has $223.44/137.56 \doteq 1.62$
times the variance for $\hat\beta_3$ compared to
running the tie-breaker on the second $10$\%. Put another way, the tie-breaker
design reduces $\var(\hat\beta_3)$ by a factor of roughly $0.6$.

In a setting like this we might find after gathering the data that
the working linear model fits poorly over the whole range of $x$
and the model would then have severe bias.
An alternative is to just analyze the top $30$\% of subjects.
In that case, the skewed RDD becomes a usual RDD and the
tie-breaker becomes an experiment with $\Delta = 1/3$.
Working on only $30$\% as many observations
over a narrower range of $x$ values will increase the
variance for both models.
The efficiency of this tie-breaker compared to the RDD is
$1+3(1/3)^2(2-(1/3)^2)\doteq 1.63$, almost identical to
what we find for the designs in the table.

\section{Discussion}\label{sec:discussion}

In an incentive plan, a regression discontinuity design rewards the a
priori best customers but it has severe disadvantages if one wants to
follow up with regression models to measure impact.
There is a tradeoff between estimation efficiency and allocation
efficiency. Proposition~\ref{prop:bestdelta} provides a principled
way to translate estimates or educated guesses about the present
value of the incentives and future value of information into a
choice of $\Delta$ in a hybrid experiment.

In commercial settings, the incentive under study will change
over time. Experience with similar though perhaps
not identical prior incentive plans then gives
some guidance for making the tradeoff.
A simpler approach
is to do the smallest experiment with at least some given
fraction of the information from $\Delta=1$.

We have examined a simple linear model because it is easiest to work
with and is a reasonable design choice in many contexts.
Analysts have many more models at their disposal when the
data come in and they do not need to use that model.
If a more satisfactory model is found then the methods
of Section~\ref{sec:general} can be used to design the next experiment under that model.
Section~\ref{sec:quadratic} on the quadratic model
provides a warning: the RDD becomes very unreliable already with
this model which is only slightly more complicated than the two-line
model. A tie-breaker greatly reduces the variance compared to RDD.

In some applications, the assignment variable may be the output
of a scoring model based on many subject variables.
We expect that incorporating randomness into the design will
give better data for refitting such an underlying scoring model,
but following up that point is outside the scope of this article.
The effects are likely to vary considerably from problem to problem.

%\section*{Acknowledgments}
%We thank Eric Tassone for discussions.

%\bibliographystyle{apalike}
\bibliographystyle{plainnat}
\bibliography{rd}

\begin{thebibliography}{27}
\providecommand{\natexlab}[1]{#1}
\providecommand{\url}[1]{\texttt{#1}}
\expandafter\ifx\csname urlstyle\endcsname\relax
  \providecommand{\doi}[1]{doi: #1}\else
  \providecommand{\doi}{doi: \begingroup \urlstyle{rm}\Url}\fi

\bibitem[Abdulkadiroglu et~al.(2017)Abdulkadiroglu, Angrist, Narita, and
  Pathak]{abdulkadiroglu2017impact}
Atila Abdulkadiroglu, Joshua~D Angrist, Yusuke Narita, and Parag~A Pathak.
\newblock Impact evaluation in matching markets with general tie-breaking.
\newblock Technical report, National Bureau of Economic Research, 2017.
\newblock URL \url{http://www.nber.org/papers/w24172}.

\bibitem[Aiken et~al.(1998)Aiken, West, Schwalm, Carroll, and
  Hsiung]{aike:west:schw:carr:hsiu:1998}
Leona~S Aiken, Stephen~G West, David~E Schwalm, James~L Carroll, and Shenghwa
  Hsiung.
\newblock Comparison of a randomized and two quasi-experimental designs in a
  single outcome evaluation: Efficacy of a university-level remedial writing
  program.
\newblock \emph{Evaluation Review}, 22\penalty0 (2):\penalty0 207--244, 1998.

\bibitem[Angrist et~al.(2014)Angrist, Hudson, and Pallais]{Angrist2014leveling}
Joshua Angrist, Sally Hudson, and Amanda Pallais.
\newblock Leveling up: Early results from a randomized evaluation of
  post-secondary aid.
\newblock Technical report, National Bureau of Economic Research, 2014.
\newblock URL \url{http://www.nber.org/papers/w20800.pdf}.

\bibitem[Angrist and Pischke(2009)]{Angrist09}
Joshua~D. Angrist and Jorn-Steffen Pischke.
\newblock \emph{Mostly Harmless Econometrics}.
\newblock Princeton Univerity Press, Princeton, 2009.

\bibitem[Angrist and Pischke(2014)]{Angrist14}
Joshua~D. Angrist and Jorn-Steffen Pischke.
\newblock \emph{Mastering Metrics}.
\newblock Princeton Univerity Press, Princeton, 2014.

\bibitem[Armstrong and Koles{\'a}r(2018)]{armstrong2018optimal}
Timothy~B Armstrong and Michal Koles{\'a}r.
\newblock Optimal inference in a class of regression models.
\newblock \emph{Econometrica}, 86\penalty0 (2):\penalty0 655--683, 2018.

\bibitem[Box et~al.(1978)Box, Hunter, and Hunter]{box1978statistics}
George E.~P. Box, William~Gordon Hunter, and J.~Stuart Hunter.
\newblock \emph{Statistics for experimenters}.
\newblock John Wiley and Sons, New York, 1978.

\bibitem[Calonico et~al.(2014)Calonico, Cattaneo, and
  Titiunik]{calonico2014robust}
Sebastian Calonico, Matias~D Cattaneo, and Rocio Titiunik.
\newblock Robust nonparametric confidence intervals for
  regression-discontinuity designs.
\newblock \emph{Econometrica}, 82\penalty0 (6):\penalty0 2295--2326, 2014.

\bibitem[Campbell(1969)]{camp:1969}
Donald~T Campbell.
\newblock Reforms as experiments.
\newblock \emph{American psychologist}, 24\penalty0 (4):\penalty0 409, 1969.

\bibitem[Cappelleri and Trochim(2003)]{CappelleriXX}
Joseph~C. Cappelleri and William M.~K. Trochim.
\newblock Cutoff designs.
\newblock In Marcel Dekker, editor, \emph{Encyclopedia of Biopharmaceutical
  Statistics}. CRC Press, 2003.
\newblock \doi{10.1081/E-EBS 12000734}.
\newblock URL
  \url{https://www.socialresearchmethods.net/research/Cutoff%20Designs%202003.pdf}.

\bibitem[Gelman and Imbens(2017)]{gelman2017high}
Andrew Gelman and Guido Imbens.
\newblock Why high-order polynomials should not be used in regression
  discontinuity designs.
\newblock \emph{Journal of Business \& Economic Statistics}, 0\penalty0 (0),
  2017.
\newblock URL \url{http://www.nber.org/papers/w20405}.

\bibitem[Goldberger(1972)]{gold:1972}
A.~S. Goldberger.
\newblock Selection bias in evaluating treatment effects: Some formal
  illustrations.
\newblock Technical Report Discussion paper 128--72, Institute for Research on
  Poverty, University of Wisconsin--Madison, 1972.

\bibitem[Hahn et~al.(2001)Hahn, Todd, and der Klaauw]{hahn2001identification}
Jinyong Hahn, Petra Todd, and Wilbert~Van der Klaauw.
\newblock Identification and estimation of treatment effects with a
  regression-discontinuity design.
\newblock \emph{Econometrica}, 69\penalty0 (1):\penalty0 201--209, 2001.

\bibitem[Imbens and Lemieux(2008)]{Imbens08}
Guido Imbens and Thomas Lemieux.
\newblock Regression discontinuity designs: a guide to practice.
\newblock \emph{Journal of Econometrics}, 142\penalty0 (2):\penalty0 615--635,
  2008.
\newblock URL \url{www.nber.org/papers/w13039.pdf}.

\bibitem[Imbens and Wager(2019)]{imbens2019optimized}
Guido Imbens and Stefan Wager.
\newblock Optimized regression discontinuity designs.
\newblock \emph{Review of Economics and Statistics}, 101\penalty0 (2):\penalty0
  264--278, 2019.

\bibitem[Imbens and Rubin(2015)]{imbens2015causal}
Guido~W Imbens and Donald~B Rubin.
\newblock \emph{Causal inference in statistics, social, and biomedical
  sciences}.
\newblock Cambridge University Press, 2015.

\bibitem[Jacob et~al.(2012)Jacob, Zhu, Marie-Andrée, Somers, and
  Bloom]{Bloom08}
Robin~Tepper Jacob, Pei Zhu, Marie-Andrée, Somers, and Howard Bloom.
\newblock A practical guide to regression discontinuity.
\newblock \emph{MDRC Publications}, July 2012.
\newblock URL
  \url{https://www.mdrc.org/publication/practical-guide-regression-discontinuity}.

\bibitem[Lee and Lemieux(2010)]{Lee10}
David~S. Lee and Thomas Lemieux.
\newblock Regression discontinuity designs in economics.
\newblock \emph{Journal of Economic Literature}, 48:\penalty0 281--355, June
  2010.
\newblock URL \url{https://www.princeton.edu/~davidlee/wp/RDDEconomics.pdf}.

\bibitem[McCrary(2008)]{mccrary2008manipulation}
Justin McCrary.
\newblock Manipulation of the running variable in the regression discontinuity
  design: A density test.
\newblock \emph{Journal of econometrics}, 142\penalty0 (2):\penalty0 698--714,
  2008.

\bibitem[Porter(2003)]{porter2003estimation}
Jack Porter.
\newblock Estimation in the regression discontinuity model.
\newblock \emph{Unpublished Manuscript, Department of Economics, University of
  Wisconsin at Madison}, 2003:\penalty0 5--19, 2003.

\bibitem[Rice and Rosenblatt(1983)]{rice1983smoothing}
John Rice and Murray Rosenblatt.
\newblock Smoothing splines: regression, derivatives and deconvolution.
\newblock \emph{The annals of Statistics}, pages 141--156, 1983.

\bibitem[Rosenman and Rajkumar(2019)]{rosenman2019optimized}
Evan Rosenman and Karthik Rajkumar.
\newblock Optimized partial identification bounds for regression discontinuity
  designs with manipulation.
\newblock Technical Report arXiv:1910.02170, Stanford University, 2019.

\bibitem[Scott(2015)]{scott2015multi}
Steven~L. Scott.
\newblock Multi-armed bandit experiments in the online service economy.
\newblock \emph{Applied Stochastic Models in Business and Industry},
  31\penalty0 (1):\penalty0 37--45, 2015.

\bibitem[Student(1931)]{stud:1931}
Student.
\newblock The {Lanarkshire} milk experiment.
\newblock \emph{Biometrika}, 23\penalty0 (2/3):\penalty0 398--406, 1931.

\bibitem[Thistlethwaite and Campbell(1960)]{this:camp:1960}
D.~L. Thistlethwaite and D.~T. Campbell.
\newblock Regression-discontinuity analysis: An alternative to the ex post
  facto experiment.
\newblock \emph{Journal of Educational psychology}, 51\penalty0 (6):\penalty0
  309, 1960.

\bibitem[{Van Der Klaauw}(2008)]{Kaauw08}
Wilbert {Van Der Klaauw}.
\newblock Regression--discontinuity analysis: {A} survey of recent developments
  in economics.
\newblock \emph{LABOUR}, 22\penalty0 (2):\penalty0 219--245, 2008.
\newblock URL
  \url{http://citeseerx.ist.psu.edu/viewdoc/download?doi=10.1.1.466.956&rep=rep1&type=pdf}.

\bibitem[Wu and Hamada(2011)]{wu2011experiments}
C.~F.~Jeff Wu and Michael~S. Hamada.
\newblock \emph{Experiments: planning, analysis, and optimization}.
\newblock John Wiley \& Sons, New York, 2011.

\end{thebibliography}
\vfill\eject

\end{document}